\documentclass[runningheads]{llncs}
\pdfoutput=1

\usepackage{amsmath}

\newcommand{\ks}{\mathcal{S}}
\newcommand{\app}{\mathcal{W}}
\newcommand{\appks}{\mathcal{A}}

\title{Extending FO(ID) with Knowledge Producing Definitions: Preliminary Results}
\author{Joost Vennekens\inst{1} \and Marc Denecker\inst{2}}

\institute{
Campus De Nayer, Lessius Mechelen, KU Leuven\\
De Nayerlaan 5\\
2860 Sint-Katelijne-Waver, Belgium
\and 
Dept. Computerscience, KU leuven\\
Celestijnenlaan 200A\\
3001 Heverlee, Belgium
}
\titlerunning{Extending FO(ID) with knowledge producing definitions: preliminary results}
\authorrunning{J.~Vennekens and M.~Denecker}

\begin{document}

\setcounter{page}{161}

\maketitle

\begin{abstract}
Previous research into the relation between ASP and classical logic has identified at least two different ways in which the former extends the latter.  First, ASP program typically contain sets of rules that can be naturally interpreted as inductive definitions, and the language FO(ID) has shown that such inductive definitions can elegantly be added to classical logic in a modular way.  Second, there is of course also the well-known epistemic component of ASP, which was mainly emphasized in the early papers on stable model semantics.  To investigate whether this kind of knowledge can also, and in a similarly modular way, be added to classical logic, the language of Ordered Epistemic Logic was presented in recent work.  However, this logic views the epistemic component as entirely separate from the inductive definition component, thus ignoring any possible interplay between the two.  In this paper, we present a language that extends the inductive definition construct found in FO(ID) with an epistemic component, making such interplay possible. The eventual goal of this work is to discover whether it is really appropriate to view the epistemic component and the inductive definition component of ASP as two separate extensions of classical logic, or whether there is also something of importance in the combination of the two.
\end{abstract}

\section{Introduction}

Today, Answer Set Programming (ASP) is a vibrant domain, boasting both mature technologies and successful real-world applications. The roots of ASP lie in the fields of Logic Programming (LP) and Non-monotonic Reasoning (NMR).  Both of these were initially motivated by dissatisfaction with classical first-order logic (FO), be it its computational properties (in the case of LP) or its suitability for representing common-sense knowledge (in the case of NMR). The current success of ASP suggests that, to a large extent, this domain was indeed able to overcome these problems of classical logic. 

It is, however, not yet quite clear {\em how} precisely this was done.  That is to say, the relation between ASP and classical logic is, in our opinion, not yet fully understood. Currently, ASP still stands as an {\em alternative to} FO: to effectively write ASP programs, one basically has to leave behind all methodologies, tools and intuitive understandings of classical logic and start anew in a different setting. This paper is part of a research project that attempts to close this gap \cite{MG65/DeneckerVVWB10}.  The aim is to investigate whether and how the achievements of ASP can be reformulated as modular improvements or extensions of classical logic.  Ultimately, we would like to be able to characterize ASP as a set of specific solutions to a number of orthogonal problems/limitations of classical logic, such that  someone working in classical logic could add as many or as few ``ASP-style features'' to his knowledge base as is needed for that particular application.  Of course, the motivation for this research is not purely practical. By reformulating the contributions of ASP in the classical framework, we also hope to provide a synthesis that will eventually lead to an increased understanding of classical and computational logic, and their role in problem solving.

Ironically, ASP's relation to classical logic seems currently best understood when it comes to computational aspects.  For instance, \cite{MitchellT05} showed that the typical ASP way of encoding search problems can be captured quite elegantly in a classical context by the notion of {\em modal expansion}: given a theory $T$ in an alphabet $\Sigma$ and an interpretation $I_o$ for some subvocabulary $\Sigma_o \subseteq \Sigma$, find an interpretation $I$ that extends $I_o$ to the entire vocabulary such that $I \models T$. Indeed, the 2011 edition of the ASP-competition \cite{lpnmr/DeneckerVBGT09} has had two modal expansion systems for (extensions of) classical logic among its competitors: Enfragmo \cite{aavani12} and IDP \cite{lash06/MarienWD06}. On a more technical level, also the similarities between current ASP solvers and SAT solvers are of course well understood (e.g., \cite{Lierler11,gebser,Giunchiglia}).

When it comes to the knowledge representation properties of ASP (i.e., the intuitive meaning of expressions and the modeling methodologies that have to be followed), the relation to classical logic is less clear.  As pointed out by \cite{Denecker12}, one of the key problems lies in the interpretation of the semantic structures: whereas an answer set in ASP is traditionally interpreted in an epistemic way, as a representation of the knowledge of some rational agent, classical logic is based on the Tarskian view of a model as a representation of a possible objective state of affairs.  Nevertheless, a series of papers by Denecker et al.~\cite{tocl/DeneckerBM01,lpnmr/DeneckerV07,tocl/DeneckerT08} has shown that a substantial portion of ASP can be understood as a combination of classical FO axioms and {\em inductive definitions}. 

An inductive definitions is a well-understood mathematical construct, that is usually represented by a set of natural language {\em if-then}-statements.  As shown in \cite{tocl/DeneckerBM01,lpnmr/DeneckerV07,tocl/DeneckerT08}, we can view a set of normal logic programming rules as a formal representation of such an inductive definition.  For instance, the following pair of rules:
\[
\left\{
\begin{aligned}
T(x,y) &\leftarrow T(x,z) \land T(z,y).\\
T(x,y) &\leftarrow E(x,y).
\end{aligned}
\right\}
\]
define $T$ as the transitive closure of $E$. This is of course not overly surprising and it indeed goes back to the views of Van Gelder \cite{VanGelder93} and Clark \cite{adbt/Clark78}.  Nevertheless, taking this observation seriously immediately suggests a clean and well-defined ``plugin'' that can modularly add an ASP-style component to a classical FO theory.  Indeed, because an inductive definition is nothing more than a generalization of the way in which relations are usually (non-inductively)  defined by means of an FO equivalence, there should be nothing problematic (either conceptually or mathematically) about allowing such ``sets of rules that form an inductive definition''
 anywhere one is allowed to use an equivalence. FO(ID) (previously known as ID-logic) is the language that does precisely this: it extends FO, in a completely modular way, with a rule-based representation for inductive definition \cite{Denecker:CL2000}.  Representing a search problem as a model expansion problem in FO(ID) often yields results that are almost identical to a Generate-Define-Test program in ASP \cite{Lifschitz/AI02}, apart from minor syntactic details \cite{jelia/MarienGD04,Denecker12}. In this way, FO(ID) therefore fits nicely into our stated goal, by identifying one concrete way in which ASP improves upon FO and packing this up in a language construct that can be added to an existing FO knowledge base at will.

While FO(ID) seems able to naturally represent already a surprisingly large part of existing ASP practice, it by no means covers everything.  One class of examples that remains out of scope is that of the epistemic examples that originally motivated the stable model semantics \cite{iclp/GelfondL88,GelfondL90}. This is epitomized by the well-known interview example \cite{GelfondL91}, which is expressed in ASP as:
\begin{align*}
Eligible(x) &\leftarrow HighGPA(x).\\
Eligible(x) &\leftarrow Minority(x), FairGPA(x).\\
\lnot Eligible(x) & \leftarrow \lnot FairGPA(x).\\
Interview(x) & \leftarrow not~Eligible(x), not~\lnot Eligible(x).
\end{align*}
Recent efforts have attempted to reformulate also this kind of example in a way that explains its relation to FO.  In \cite{VlaeminckVBD/KR2012}, the language of {\em Ordered Epistemic Logic} (OEL) is developed for this.   Here, an ordered set of FO theories is considered, and each theory $T$ is allowed to make use of modal operators $K_{T'}$ that refer to the knowledge entailed by a theory $T' < T$.  The interview example would consist of two theories $T_1,T_2$ with $T_1 < T_2$, such that $T_1$ is a normal FO knowledge base containing facts about $Minority$, $HighGPA$ and $FairGPA$, and a definition of $Eligible$. The theory $T_2$ consists of the single equivalence:
\[
\forall x\ Interview(x) \Leftrightarrow \lnot K_{T_1} Eligible(x) \land \lnot K_{T_1} Eligible(x).
\]
This logic extends FO in a way which is completely orthogonal to FO(ID). There is nothing to prevent the different knowledge bases of an OEL theory from containing, in addition to regular FO formulas, also inductive definitions, but if they do, there is no interplay with the epistemic operators.  In other words, FO(ID) and OEL both isolate one particular non-classical aspect of ASP, and  show how it can be modularly added to FO, but they do so independently.  This presupposes that there is nothing of importance in ASP's {\em combination} of epistemic and ``inductive definition'' reasoning.  However, it does not seem {\em a priori} obvious that this is the case.

In this paper, we will therefore investigate how an epistemic component can be added to the inductive definition construct of FO(ID) itself. The key idea here is to allow both a relation {\em and} an agent's knowledge about this relation to be defined together in a single {\em knowledge producing definition}, as we will call it.  The semantics of such a knowledge producing definition is defined by a constructive process that creates, in parallel, the relations that are being defined and the agent's knowledge about them.  In this way, we obtain a language in which, unlike the FO(ID)+OEL approach, interaction between the epistemic and definitional component is possible.  The hope is that such a language might shed more light on the epistemic component of ASP and its relation to both classical FO and the inductive definitions of FO(ID).

The work presented in this paper is still at a preliminary stage, but we will attempt to sketch interesting avenues for future research.  Because one of our main design goals is to make our approximate knowledge structures integrate seamlessly with the inductive definition construct as it already exists in FO(ID), we will need to spend some time recalling the details of this, before we can develop our extension.

\section{Preliminaries: the semantics of inductive definitions}

Inductive definitions are constructive. In mathematical texts, they are typically represented by a set of {\em if-then} statements, which may be applied to add new elements to the relation(s) that is (or {\em are}, in the case of a definition by simultaneous induction) being defined.  The formal representation of such a definition in FO(ID) is by a set of rules of the form 
\begin{equation}\label{defrule}
\forall \vec{x}\ P(\vec{t}) \leftarrow \phi,
\end{equation}
where $P(\vec{t})$ is an atom and $\phi$ an FO formula.For monotone definitions, the relation being defined is simply the least relation that is closed under application of the rules, and it can be constructed by exhaustively applying them.  For non-monotone definitions, the relation being defined is no longer the least relation closed under the rules, and there may, in fact, be many minimal relations closed under the rules instead of a single least one.  In mathematical texts, such a non-monotone definition is always accompanied by a well-founded order over which the induction is performed.  While the characterization as a least set breaks down, the constructive characterization still works: the defined relation can still be constructed by repeated application of the rules, provided that these rules are applied in an order that respects the given well-founded order of induction.  For instance, the standard definition of satisfaction in classical logic is a definition over the subformula order, which means that we may only apply a rule that derives that $I \models \phi$ for some $\phi$ {\em after} all rules that could derive $I \models \psi$ for a subformula $\psi$ of $\phi$ have been applied. 

Of course, for a correct inductive definition, it is important that the structure of the rules also respects the well-founded order.  For instance, in an definition over the subformula order, it makes no sense for a rule to define whether a formula is satisfied in terms of the satisfaction of a {\em larger} formula.  Therefore, the structure of the rules and the well-founded order are not independent.  In fact, the well-founded order is already entirely implicit in the structure of the rules!  As shown in \cite{lpnmr/DeneckerV07} and a series of prior papers, the well-founded semantics (WFS) \cite{GelderRS91} can actually  be seen as a mathematical construct to recover the well-founded order from the structure of the rules.

For simplicity, in the technical material of this paper, we will restrict attention to {\em ground} formulas only.

Since its original definition, a number of alternative ways of defining the WFS have been developed.  One of these is to start from the following method of evaluating a formula $\phi$ in a pair of interpretations $(I,J)$: 
\begin{itemize}
\item For an atom $P(\vec{t})$, $(I,J)\models P(\vec{t})$ iff $I \models P(\vec{t})$,
\item For a formula $\lnot \psi$, $(I,J) \models  \lnot\psi$ iff $(J,I)\not \models  \psi$,
\item For a formula $\psi_1 \lor\psi_2$, $(I,J) \models  \psi_1 \lor\psi_2$ iff $(I,J) \models  \psi_1$ or $(I,J) \models  \psi_2$,
\item For a formula $\psi_1 \land\psi_2$, $(I,J) \models  \psi_1 \land\psi_2$ iff $(I,J) \models  \psi_1$ and $(I,J) \models  \psi_2$.
\end{itemize}
The crux of this definition lies in the case for negation, which switches the roles of $I$ and $J$, thus ensuring that positive occurrences of atoms are evaluated in $I$ and negative occurrences in $J$.  We will call a pair $(I,J)$ for which $I \leq J$ an {\em approximating pair}, because it can be seen as an approximation of the set of interpretations $K$ such that $I \leq K$ and $K \leq J$.  Indeed, if $(I,J) \models \phi$, according to the above definition, then $K \models \phi$ for all such $K$.  Moreover, if $K \models \phi$ for at least one such $K$, then $(J,I) \models \phi$.  For pairs $(I,J)$ such that $I=J$, the evaluation $(I,J)\models\phi$ reduces to classical satisfaction $I \models \phi$.  Pairs for which this is the case are called {\em total}. 

The WFS can then be defined as the unique limit of a sequence of pairs of interpretations $(I_i,J_i)_{i \geq 0}$. This sequence starts from the least precise pair of interpretations $(\bot_\Sigma,\top_\Sigma)$, where $\bot_\Sigma$ is the interpretation in which all atoms is false and $\top_\Sigma$ the interpretation in which all atoms are true. There are then two acceptable ways of going from $(I_i,J_i)$ to $(I_{i+1},J_{i+1})$:
\begin{itemize}
\item Either $J_{i+1} = J_i$ and $I_{i+1}$ is the union $I_i + \{P(\vec{t})\}$, with $P(\vec{t})$ an atom for which there exists a rule of the form $P(\vec{t})\leftarrow \phi$ with $(I_i,J_i)\models \phi$; 
\item Or $I_{i+1} = I_i$ and $J_{i+1}$ is such that $I_i \leq J_{i+1} \leq J_i$ and for all atoms $P(\vec{t})$ in the set difference $J_i - J_{i+1}$ and all rules $P(\vec{t})\leftarrow \phi$, $(J_{i+1},I_i)\not \models \phi$ 
\end{itemize}
Intuitively, the first of these two cases allows us to derive the head of a rule once it is certain that its body is satisfied (in the sense that $K\models \phi$ for all $K$ approximated by $(I_i,J_i)$).  The second case allows us to assume that a set of atoms must all be false if this assumption would make us certain that all bodies $\phi$ of rules with one of these atoms in the head are false (i.e., $K \not \models \phi$ for all $K$ approximated by $(I_i,J_{i+1})$). The set $J_{i+1}-J_i$ of atoms that are falsified in this operation is known as an {\em unfounded set}.

A sequence constructed in this way is called an {\em induction sequence}. The well-founded model  (WFM) is now precisely the unique limit $(V,W)$ to which all such induction sequences converge.  If this WFM is total (i.e., $V=W$), then the definition completely determines the extension of the predicates it defines.  Clearly, this is a desirable property for an inductive definition. Therefore, FO(ID) allows only total models.

So far, we have tacitly assumed that all predicates in the vocabulary $\Sigma$ are defined by the definition. In mathematics, however, this is rarely the case, since most definitions serve to define some relation(s) {\em in terms of} some other relation(s).  This is also possible in FO(ID).  For a definition $\Delta$ (i.e., a set of rules of form \eqref{defrule}), the predicate symbols appearing in the head of at least one of these rules are called the {\em defined predicates} of $\Delta$. The set of all defined predicates is denoted as $Def(\Delta)$.  The remaining predicates (i.e., those that belong to $\Sigma - Def(\Delta)$) are called {\em open}, and the set of all such predicates is denoted by $Op(\Delta)$.  The purpose of a definition is then to characterize the defined predicates $Def(\Delta)$ in terms of the open predicates $Op(\Delta)$.  Formally, this is done by parametrizing the construction process by an interpretation for the open predicates: for an interpretation $O$ of $Op(\Delta)$, an induction sequence {\em given $O$} is defined as a sequence of interpretation $(I_i,J_i)_{i \geq 0}$, in which all interpretations $I_i$ and $J_i$ extend the given interpretation $O$.  The starting point of this sequence is the pair $(I_0,J_0)$ such that $I_0 = O + \bot_{Def(\Delta)}$ and $J_0 = O + \top_{Def(\Delta)}$. 

An interpretation $I$ is then called a {\em model} of a definition $\Delta$, denoted $I \models \Delta$, if the unique limit of the induction sequences for $\Delta$ given $I\rvert_{Op(\Delta)}$ (i.e., the restriction of $I$ to the open predicates of $\Delta$) is precisely the total pair $(I,I)$. 

FO(ID) now consists of classical first-order logic FO extended with these definitions.  While some versions of this logic allow boolean combinations of classical formulas and inductive definitions, we will, for simplicity, restrict attention in this paper to FO(ID) theories  that consist of precisely one FO formula $\phi$ and one inductive definition $\Delta$. For such a theory $T = \{\phi,\Delta\}$, we define that $I$ is a {\em model} of $T$, denoted $I \models T$, iff both $I \models \phi$ (in the classical sense) and $I \models \Delta$ (as defined above). 

\section{Knowledge producing definitions}

The goal of this paper is take the concept of an inductive definition as it exists in FO(ID) and extend it by allowing definitions that not only define the objective extension of their defined predicates, but, at the same time, also define a rational agent's knowledge about these predicates.  A {\em modal literal} is a formula of the form $K\psi$ where $\psi$ is an FO formula.  By FO(K), we denote the language that extends FO by allowing modal literals to appear anywhere an atom $P(\vec{t})$ may appear.  Note that FO(K) therefore does not allow nesting of the operator $K$.  A knowledge producing definition $\kappa$ is a set of rules of either the form
\begin{equation}\label{f1}
\forall x\ K \psi \leftarrow \phi,
\end{equation}
or
\begin{equation}\label{f2}
\forall x\ P(\vec{t}) \leftarrow\phi.
\end{equation}
Here, $K \psi$ is a modal literal and $P(\vec{t})$ an atom.  In both cases, $\phi$ is an FO(K) formula. Again, in our formal treatment, we will always assume that these rules have already been appropriately grounded. The defined predicates of a knowledge producing definition $\kappa$ are all the predicates $P$ that appear in the head of a rule of form \eqref{f2}.  All other predicates, including those that appear only in the formula $\psi$ of a rule of form \eqref{f1}, are open.

For this formalism, the basic semantic structure will consist of a pair of an interpretation $I$, representing the real world, and a set of interpretations $W$, representing the agent's knowledge about the world.  We call such a pair $(I,W)$ a {\em knowledge structure} and call it {\em consistent} if $I \in W$. It is obvious how to evaluate a knowledge formula in such a knowledge structure.
\begin{definition}\label{def-one}
For a knowledge formula $\phi$ and a knowledge structure $\ks = (I,W)$, we define $\ks\models \phi$ as follows:
\begin{itemize}
\item For an atom $P(\vec{t})$, $\ks \models P(\vec{t})$ iff $I \models P(\vec{t})$,
\item For a modal literal $K\psi$, $\ks \models K\psi$ iff for each $J$ in $W$, $(J,W)\models \psi$,
\item The other cases are defined as usual.
\end{itemize}
\end{definition}

Like the WFS, which construct a single interpretation $K$ through a series of increasingly precise approximations by pairs of interpretations $(I_i,J_i)$, our semantics will construct a single knowledge structure $(I,W)$ through a series of increasingly precise approximations of it.  Part of these approximations will again be  a pair of interpretations $(I_i,J_i)$, which form an increasingly precise sequence of approximations of the real extension $K$ of the defined predicates.  At each stage of this approximating sequence, we will also keep track of the agent's knowledge about it.  Therefore, at each step $i$, we will also have a set $\app_i$ of pairs of interpretations; if $(I',J') \in \app_i$, then this means that the agent considers it possible that $(I',J')$ occurs somewhere in the real approximating sequence.  The way in which we will ensure this property, is to, on the one hand, apply the same derivation rules we apply to the real sequence to the pairs that the agent considers possible. On the other hand, when the agent's knowledge increases due to a modal literal in the head of a rule, this will eliminate some of the possibilities (i.e., some of these pairs $(I',J')$ are removed from $\app_i$), but it will not change any of the possibilities themselves (i.e., the remaining pairs do not change).

The following will serve as our basic semantic structure.
\begin{definition}
An {\em approximate knowledge structure} $\appks$ is a pair $((I,J), \app)$ of an approximating pair $(I,J)$ and a set $\app$ of approximating pairs $(I',J')$.
\end{definition}

Evaluating an FO(K) formula in such an approximate knowledge structure is again a matter of switching the approximating pairs when negation is encountered. Formally, we define that, for an approximate knowledge structure $\appks = ((I,J), \app)$,
\begin{itemize}
\item For an atom $P(\vec{t})$, $\appks \models P(\vec{t})$ iff $I \models P(\vec{t})$;
\item For a modal literal $K \psi$, $\appks \models K \psi$ iff for each $(I',J')\in\app$, $((I',J'), \app) \models \psi$;
\item For a formula $\lnot \phi$, $\appks\models \lnot \phi$ iff $((J,I),\overline{\app})\models \phi$, where $\overline{\app} = \{(J',I') \mid (I',J')\in \app\}$;
\item The other cases are defined as usual.
\end{itemize}

\newcounter{opcnt}
\stepcounter{opcnt}
\newenvironment{op}{\begin{description}\item[{Operation \theopcnt}.]}{\end{description}\stepcounter{opcnt}}

We now construct an increasingly precise sequence $(\appks_i)_{i\geq 0}$ of approximate knowledge structures.  If we project this sequence unto the approximating pair $(I,J)$ of each approximate knowledge structures $((I,J), \app)$, the result will essentially be just a regular induction sequence. The sequence again takes as input a knowledge structure $(O,\mathcal{M})$ for the open predicates, and its starting point is then the approximate knowledge structure
\[
\appks_0 = ( (O+\bot_{Def(\kappa)}, O+\top_{Def(\kappa)}), \{ (O'+\bot_{Def(\kappa)},O'+\top_{Def(\kappa)}) \mid O' \in \mathcal{M}\}).\]
We construct subsequent elements of the sequence by applying one of the following operations.
\begin{op}
$\appks_{i+1} = ((I_i + \{P(\vec{t})\}, J_i), \app_i)$ where $\appks_i =  ((I_i, J_i), \app_i)$ such that there is a rule $r$ of the form $P(\vec{t})\leftarrow \phi$ with $((I_i,J_i),\app_i) \models \phi$. 
\end{op}
This operation is just the obvious analogue to the first operation used in building normal induction sequences. The only difference is that the agent's knowledge $\app_i$ is dragged along as an additional argument, which is used to evaluate occurrences of modal literals in the rule bodies.
\begin{op}
$\appks_{i+1} = ((I_i, J_i), \app_{i+1})$ where $\appks_i =  ((I_i, J_i), \app_i)$ and $\app_{i+1} = \app_i - \{(I_i',J_i')\} + \{(I_{i}' + \{P(\vec{t})\},J_i)\}$ such that there is a rule $r$ of the form $P(\vec{t})\leftarrow \phi$ with $((I_i',J_i'),\app_i)\models \phi$.
\end{op}
This operation is essentially the same as the previous one, with the only difference being that it is now not applied to the approximating pair $(I_i,J_i)$, but to one of the approximating pairs $(I_i',J_i')$ in $\app_i$. Where the previous two operations are analogous to the production operation of the normal induction sequence, the next two mimic the unfounded set operation. 
\begin{op} $\appks_{i+1} = ((I_i, J_{i+1}), \app_i)$ where $\appks_i =  ((I_i, J_i), \app_i)$ 
and $J_{i+1}$ is such that $I_i \leq J_{i+1} \leq J_i$ and for all atoms $P(\vec{t}) \in J_i - J_{i+1}$ and all rules $r$ of the form $P(\vec{t})\leftarrow \phi$, it holds that $(J_{i+1},I_i)\not \models \phi$. 
\end{op}
Again, this operation can either be performed  on the approximating pair $(I_i,J_i)$, as above, or on one of the pairs in the set $\app_i$, as below:
\begin{op}
Or $\appks_{i+1} = ((I_i, J_i), \app_{i+1})$ where $\appks_i =  ((I_i, J_i), \app_i)$ and $\app_{i+1} = \app_i - \{(I_i',J_i')\} + \{(I_{i}',J_{i+1}')\}$  and  $J_{i+1}'$ is such that $I_i' \leq J_{i+1}' \leq J_i'$ and for all atoms $P(\vec{t}) \in J_i' - J_{i+1}'$ and all rules $r$ of the form $P(\vec{t})\leftarrow \phi$, it holds that $((J_{i+1}',I_i'),\app_i) \not \models \phi$. 
\end{op}
The final operation takes care of the effect of knowledge producing rules. 
\begin{op}
$\appks_{i+1} = ((I_i, J_i), \app_{i+1})$ where $\appks_i =  ((I_i, J_i), \app_i)$ and $\app_{i+1} = \{ (I,J) \in \app_i \mid (J,I) \models \psi\}$ and there exists a rule $r$ of the form $K\psi \leftarrow \phi$ such that $((I_i, J_i), \app_{i})\models \phi$. 
\end{op}

Note that the condition for removing a pair $(I,J)$ from $\app_i$ is that $(J,I) \not \models \psi$, i.e., that no interpretation approximated by $(I,J)$ still satisfies $\psi$.  

The semantics of a knowledge producing definition is now defined in terms of sequences that are constructed by these operations.

\begin{definition} Let $\kappa$ by a knowledge producing definition. Let $(O, \mathcal{M})$ be a knowledge structure describing the agent's initial knowledge about the open predicates.  A {\em knowledge derivation sequence} is a sequence $(\appks_i)_{i\geq 0}$ of approximate knowledge structures, starting from \[
\appks_0 = ( (O+\bot_{Def(\kappa)}, O+\top_{Def(\kappa)}), \{ (O'+\bot_{Def(\kappa)},O'+\top_{Def(\kappa)}) \mid O' \in \mathcal{M}\}),\]
such that each $\appks_{i+1}$ is obtained from $\appks_i$ by applying one of the five operations defined above and, to prevent the same operation from being applied again and again, $\appks_{i+1} \neq \appks_i$. Such a sequence is called {\em complete} if there is no way to extend it further without violating this condition.  It is called {\em sound} if every operation that is used to construct $\appks_{i+1}$ from $\appks_i$ remains applicable in all $\appks_j$ with $j > i$.  It is called {\em total} if it terminates in an approximate knowledge structure $((I,J),\app)$ such that $I=J$ and for each $(I',J')\in \app$ also $I'=J'$.
\end{definition}

The condition of totality is borrowed from FO(ID), where, as mentioned, it is used to ensure that each inductive definition correctly and completely defines the relations it sets out to define.  It is therefore also a natural requirement in our context. 

The condition of soundness is meant to avoid situations in which the sequence ends up contradicting itself, as it might for the following example: 
\[
\{K p \leftarrow \lnot K p.\}
\]
Here, the fact that the agent does not know $p$ will produce the knowledge that $p$. This is not only conceptually problematic, but also creates the practical problem that the order in which operations are applied might have an effect on the final outcome. For instance, the knowledge definition
\[
\left\{\begin{aligned}
K q \leftarrow \lnot K p.\\
K p \leftarrow \lnot K q.
\end{aligned}\right\}
\]
has both a derivation sequence that starts with the first rule and therefore ends up knowing $q$ but not knowing $p$, and one that starts with the second rule and ends up knowing $p$ but not $q$.

This can only happen with sequences that are not sound, as the following proposition shows.

\begin{proposition}
All complete and sound knowledge derivation sequences that start from the same knowledge structure $(O, \mathcal{M})$ terminate in the same approximate knowledge structure $\appks$.  
\end{proposition}
\begin{proof}[sketch]
At any particular point in the derivation sequence, many operation may be applicable. We have to show that it does not matter which of these we choose.  First, because the sequence is sound, we know that even if we do not choose to apply an operation now, we will always get a chance to apply it later. Second, the effect of an operation usually does not depend on when it is executed. The only exception to this is Operation 5, because the set of approximating pairs may of course change throughout the sequence. However, the only changes to this set are that (1) pairs are removed and that (2) pairs become more precise (i.e., that either the first element $I$ of a pair $(I,J)$ becomes larger or that $J$ becomes smaller, so that fewer interpretations $K$ lie between $I$ and $J$).  Clearly, (1) is not a problem, because if a pair gets removed later on any way, it does not matter if we already remove it now or not. Also (2) is not a problem, because the condition for removing a pair $(I,J)$, namely that $(J,I)\not\models \phi$, also implies that, for any more precise pair $(I',J')$  (i.e., such that $I \leq I'$ and $J' \leq J$), it will be the case that $(J',I')\not\models \phi$ as well. The only effect of postponing the application of an Operation 5 to a later stage is therefore that we might end up removing {\em more} pairs than if we had applied it now.  However, if we apply the operation now, then at the later stage we will have another Operation 5 available, namely the one that removes precisely those pairs that form the difference.  Since this operation will remain applicable and therefore must eventually be applied, the end result will be the same.

\qed
\end{proof}

Moreover, the only way in which it is possible to obtain unsound derivations is by negated modal literals.

\begin{proposition}
For a knowledge producing definition in which each body of a rule contain only positive occurrences of modal literals $K\psi$, each knowledge derivation sequence is sound.
\end{proposition} 
\begin{proof}[sketch]
The differences between one approximate knowledge structure $\appks_i = ((I_i,J_i),\app_i)$ and an approximate knowledge structure $\appks_j = ((I_j,J_j),\app_j)$ that occurs later in the derivation sequence are:
\begin{itemize}
\item $I_i \leq I_j$ and $J_j \leq J_i$: this implies that whenever $((I_i,J_i),\app)\models \phi$ for some $\app$ and $\phi$, also $((I_j,J_j),\app)\models \phi$
\item $\app_j$ consists of pairs $(I_j',J_j')$ for which there exists a corresponding pair $(I_i',J_i')\in\app_i$ such that, again, $I_i' \leq I_j'$ and $J_j' \leq J_i'$, and therefore whenever $((I_i',J_i'),\app) \models\phi$ also $((I_j',J_j'),\app) \models\phi$
\item for some $(I_j',J_j') \in \app_j$, there may not exist a corresponding pair in $\app_i$ 
\end{itemize}
Putting the second and third point together, it is obvious that $\app_j$ always knows everything that $\app_i$ knows, and possibly more.  Therefore, only rule bodies containing a negative occurrence of a modal literal may becomes false in $\app_j$ after they were true in $\app_i$.  

\qed
\end{proof}

This result seems to suggest that it might be possible to impose syntactic constraints on a knowledge producing definition to ensure that it only has sound derivations, by limiting the way in which negated modal literals are allowed to appear. However, it is not enough to, e.g., just require that these definitions are stratified, because gaining new knowledge about one predicate may have ``side effects'' where also knowledge about another predicate is produced.  For instance, consider the following definition which, at first sight, contains no cycles at all:
\[
\left\{
\begin{aligned}
q & \leftarrow p.\\
K q & \leftarrow r.\\
r & \leftarrow \lnot Kp.
\end{aligned} 
\right\}
\]
Even though there are no syntactic cycles, once the agent learns $q$, this will also produce the knowledge that $p$ (since $q$ is only true in worlds where $p$ also holds).  Because of the  $\lnot K p$ in the body of the rule for $r$, this will lead to an unsound derivation.

Our approach in this paper will be to ignore the existence of unsound derivations and view the unique limit of the sound derivations as the semantics of a knowledge producing definition---at least, if this unique limit is total.  If the limit is not total (given a particular knowledge structure $\ks$ for the open predicates), then, just like in FO(ID), the knowledge producing definition simply has no models (for that particular $\ks$).   If no sound derivations exists (for that $\ks$), then again the knowledge producing definition has no models (for that $\ks$).  

\section{Adding knowledge producing definitions to FO}

The previous section defined the concept of a knowledge producing definition in isolation. Of course, our goal is to add this construct to FO(K), in the same way as FO(ID) has added inductive definitions to FO.  We will again consider only theories of the form $T = \{\phi,\kappa\}$, were $\phi$ is now an FO(K) formula and $\kappa$ a knowledge producing definition.  We define that a consistent knowledge structure $\ks = (I,W)$ is a {\em weak model} of $T$ if $\ks\models \phi$ and the approximate knowledge structure $((I,I), \{ (J,J) \mid J \in W\})$ is the unique total limit of each sound derivation sequence that starts from $(I\rvert_{Op(\kappa)},\{L\rvert_{Op(\kappa)}\mid L\in W\})$.

A problem with these weak models is they might contain knowledge that is not warranted.  Consider, for instance, the following example:
\begin{equation*}
\left\{
q \leftarrow Kp
\right\}.
\end{equation*}
Here, there is no reason at all for knowing $p$, yet this knowledge producing definition has a weak model $(\{p,q\}, \{ \{p,q\}\})$.  To avoid such models, we introduce the following concept. 

\begin{definition}
Let $T = \{\phi,\kappa\}$, were $\phi$ is an FO(K) formula and $\kappa$ a knowledge producing definition.  A knowledge structure $\ks = (I,W)$ is a {\em strong model} of $T$, denoted $\ks \models T$, if $\ks\models \phi$ and the approximate knowledge structure $((I,I), \{ (J,J) \mid J \in W\})$ is the unique total limit of each sound derivation sequence that starts from $(I\rvert_{Op(\kappa)}, O)$ where $O$ is the set of all $L\rvert_{Op(\kappa)}$ such that there exists a $W'$ for which $(L,W')$ is a weak model of $T$.
\end{definition} 

By always using the set $O$ as the set of possible worlds for the open predicates, this definition prevents the knowledge producing definition from arbitrarily knowing more about its open predicates than it should.

\section{Examples}

In the logic we have now defined, it is straightforward to represent the $Interview$ example.
\begin{gather*}
\left\{
\begin{aligned}
Eligible(x) &\leftarrow (HighGPA(x) \lor (Minority(x) \land FairGPA(x)).\\
Interview(x) & \leftarrow \lnot K~Eligible(x) \land \lnot K~\lnot Eligible(x).
\end{aligned}
\right\}\\
\forall x\ HighGPA(x) \Leftrightarrow x =Mary.\\
\forall x\ FairGPA(x) \Leftrightarrow x =John.\\
Minority(Mary).
\end{gather*}
Here, the soundness condition on the derivation sequence means that we first have to apply the first rule of the knowledge producing definition to exhaustion, before starting with the second rule. There are two possible interpretations for the open predicates of the definition that satisfy the FO part of the theory:
\begin{gather*}
O_1 = \{ HighGPA(Mary) , FairGPA(John) , Minority(Mary) \},\\
O_2 = \{ HighGPA(Mary) , FairGPA(John) , Minority(Mary) , Minority(John)  \}.
\end{gather*}
Each strong model will therefore have to be produced by a derivation sequence in which $W_0 = \{ O_1,O_2\}$ is used as the agent's knowledge about the open predicates. Let us first consider a sound derivation sequence that interprets the open predicates by the knowledge structure $\ks_1 = (O_1,W_0)$. The first step is $\appks_0 = (A, \{A,B\})$ with:
\begin{align*}
A &= (O_1 + \bot_{Def(\kappa)}, O_1 +\top_{Def(\kappa)})\\ 
B &= (O_2 + \bot_{Def(\kappa)}, O_2 +\top_{Def(\kappa)}) 
\end{align*}
In the approximating pair $A$, we can apply the first rule of the knowledge producing definition to derive that $Eligible(Mary)$, and we can also apply an unfounded set operation to derive that $Eligible(John)$ is false. This leads to the new approximating pair (we abbreviate the names of the predicates):
\begin{multline*}
A' = (O_1+\{Elig(Mary)\}, O_1+\{Elig(Mary), Int(Mary),Int(John)\})
\end{multline*}
In the approximating pair $B$, on the other hand, we can derive that both $Eligible(John)$ and $Eligible(Mary)$.  This leads to the new pair:
\begin{multline*}
B' = (O_2+\{Elig(Mary),Elig(John)\}, \\O_2+\{Elig(Mary),Elig(John), Int(Mary),Int(John)\})
\end{multline*}
After applying these six operations (two times two for $A$, and two for $B$) to the approximate knowledge structure $\app_0$, we will therefore eventually end up in
\[
\app_6 = (A', \{A',B'\})
\]
Now, $\{A',B'\} \models K~Elig(Mary)$, because $Elig(Mary)$ holds in both underestimates $O_1+\{Elig(Mary)\}$ and $O_2+\{Elig(Mary),Elig(John)\}$.  Also, $\{A',B'\}\models \lnot K~Elig(John)$, since $Elig(John)$ does not hold in the overestimate (the negation switches the pairs)  of $A'$, namely $O_1+\{Elig(Mary), Int(Mary),Int(John)\}$.  Finally, $\{A',B'\}\models \lnot K~\lnot Elig(John)$ holds as well, since $Elig(John)$ does not belong to the underestimate (the two negations switch the pairs twice) of $A'$, namely $O_1+\{Elig(Mary)\}$.  Therefore, we can apply to both $A'$ and $B'$ an unfounded set operation (Operations 3 and 4) to derive that $Mary$ should not be interviewed, and we can apply the last rule of the definition to derive that $John$ should.  After six more steps (two times two for $A'$, and two for $B'$), we therefore end up in the limit
$\app_9 = (A'', \{A'', B''\})$, where
\begin{align*}
A'' =& (I,I) \text{ with } I = O_1+\{Elig(Mary), Int(John)\}\\
B'' =& (J,J) \text{ with } J = O_1+\{Elig(Mary), Elig(John), Int(John)\}
\end{align*}
This structure is total and therefore the single knowledge structure $(I,\{I,J\})$ that it approximates is a model of this theory. By a similar reasoning, there is also a second model, namely $(J,\{I,J\})$.

\subsection{Another version of the $Interview$ example} 

In the above version of this example, we are basically already encoding the solution to the problem by ordering the agent to interview everyone whose eligibility is not known.  Using a knowledge producing definition, however, it is also possible to let the semantics do more of the work.

\begin{gather*}
\left\{
\begin{aligned}
Eligible(x) &\leftarrow (HighGPA(x) \lor (Minority(x) \land FairGPA(x)).\\
K~Minority(x) & \leftarrow Interview(x) \land Minority(x).\\
K~\lnot Minority(x) & \leftarrow Interview(x) \land \lnot Minority(x).
\end{aligned}
\right\}\\
(\forall x\ HighGPA(x) \Leftrightarrow x =Mary) \land (\forall x\ FairGPA(x) \Leftrightarrow x =John) \land
Minority(Mary).\\
\forall x\ K\,Eligible(x) \lor K\,\lnot Eligible(x).
\end{gather*}

Here, we are just telling the agent that the action of interviewing a candidate will reveal his minority status, without explicitly saying who should be interviewed.  The FO(K) constraint then orders the agent to make sure that for each candidate, it is known whether he is eligible of not.  Generating models for this theory will then correctly produce plans in which people whose minority status is unknown will be interviewed.

\subsection{Sensing actions}

A successful application area of ASP is planning.  This also falls naturally in the scope of FO(ID), since theories in the situation or event calculus are essentially just an inductive definition of the values of the fluents at different points in time \cite{DeneckerT07}.  Here is an example of a theory in FO(ID) that represents a small action domain in which there is a dirty glass that can be cleaned by wiping it.

\begin{gather*}
\left\{
\begin{aligned}
Clean(t+1) & \leftarrow Wipe(t) \lor Clean(t).\\
Clean(0) &\leftarrow InitClean.
\end{aligned}
\right\}\\
\lnot InitClean
\end{gather*}

If the agent now does not know whether the glass is initially clean, we may be interested in finding a plan that will allow it to know with certainty that it will be clean at a certain  point in time. This can be accomplished by just adding, e.g., $K~Clean(2)$ as an FO(K) constraint to the theory.  More interestingly, knowledge producing definitions can also be used to add {\em sensing} actions, such as an action {\em Inspect} that allows the agent to discover whether the glass is clean.

\begin{gather*}
\left\{
\begin{aligned}
Clean(t+1) & \leftarrow Wipe(t) \lor Clean(t).\\
Clean(0) &\leftarrow InitClean.\\
K~Clean(t+1) & \leftarrow Inspect \land Clean(t).\\
K~\lnot Clean(t+1) & \leftarrow Inspect \land \lnot Clean(t).
\end{aligned}
\right\}\\
K~Clean(2).
\end{gather*}

\section{Discussion}

ASP is able to express epistemic examples by interpreting an answer set as a set of literals that are {\em believed} by a rational agent.  This is one of the most radical ways in which ASP departs from classical logic, in which models or interpretations always represent the {\em objective} state of the world.  In the classical setting, one typically resorts to sets of interpretations (or the related concept of a Kripke structure) to represent beliefs.  

In order to relate this aspect of ASP to classical logic, or to even integrate the two, it is necessary to construct a formalisation which sticks to these classical semantics objects.  In this paper, we have introduced knowledge     producing definitions for this purpose.  As our examples have shown, these are able to mimic the ASP representation of, e.g., the \textit{Interview} example, while at the same time also introducing some interesting new possibilities, such as the ability to distinguish between some atom $P(\vec{t})$ becoming objectively true (by having $P(\vec{t})$ in the head of a rule) and the agent learning this atom (by having $K~P(\vec{t})$ in the head).

\end{document}